\documentclass[11pt]{article}
\usepackage{amsfonts,amsmath,amssymb, amsthm}
\usepackage{algorithmic, algorithm, authblk}
\usepackage[margin = 1in]{geometry}
\usepackage{hyperref}

\newcommand{\C}{{\cal C}}

\newcommand{\ListLengths}{\setlength{\itemsep}{0ex}\setlength{\topsep}{1ex}\setlength{\partopsep}{0ex}}

\newenvironment{mylowitemize}{\begin{list}{$\bullet$}{\setlength{\leftmargin}{1em}\ListLengths}}{\end{list}}

\newcommand{\junk}[1]{}
\makeatletter
\newcommand\footnoteref[1]{\protected@xdef\@thefnmark{\ref{#1}}\@footnotemark}
\makeatother

\title{Plane Gossip: Approximating rumor spread in planar graphs\footnote{This material is based upon research supported in part by the U. S. Office of Naval Research under award number N00014-12-1-1001 and National Science Foundation under award number CCF-1527032.}}

\author[1]{Jennifer Iglesias\footnote{Supported by the National Science Foundation Graduate Research Fellowship Program under Grant No. 2013170941.}}
\author[2]{Rajmohan Rajaraman}
\author[1]{R. Ravi}
\author[2]{Ravi Sundaram}
\affil[1]{Carnegie Mellon University, Pittsburgh PA, USA\\ \texttt{\{jiglesia,ravi\}@andrew.cmu.edu}}
\affil[2]{Northeastern University, Boston MA, USA\\ \texttt{\{rraj,koods\}@ccs.neu.edu}}

\newcommand{\newthmwithin}[3]{\newtheorem{#1q}{#2}[#3]
                        \newenvironment{#1}{\begin{#1q}\sf}{\end{#1q}}}

\newcommand{\newthm}[3]{\newtheorem{#1q}[#2q]{#3}
                        \newenvironment{#1}{\begin{#1q}\sf}{\end{#1q}}}

\newthmwithin{theorem}{Theorem}{section}
\newthm{definition}{theorem}{Definition}
\newthm{lemma}{theorem}{Lemma}
\newthm{conjecture}{theorem}{Conjecture}
\newthm{corollary}{theorem}{Corollary}

\begin{document}

\maketitle
\thispagestyle{empty}

\begin{abstract}
We study the design of schedules for multi-commodity multicast. In this problem, we are given an
undirected graph $G$ and a collection of source-destination pairs, and
the goal is to schedule a minimum-length sequence of matchings that
connects every source with its respective destination.
Multi-commodity multicast models a classic information
dissemination problem in networks where the primary communication
constraint is the number of connections that a given node can make,
not link bandwidth. Multi-commodity multicast and its special cases,
(single-commodity) broadcast and multicast, are all NP-complete and the best
approximation factors known are $2^{\widetilde{O}(\sqrt{\log n})}$, $O(\log
n/\log\log n)$, and $O(\log k/\log\log k)$, respectively, where $n$
is the number of nodes and $k$ is the number of terminals in the
multicast instance.

Multi-commodity multicast is closely related to the problem of finding
a subgraph of optimal poise, where the poise is defined as the sum of
the maximum degree of the subgraph and the maximum distance between
any source-destination pair in the subgraph. We first show that for
any instance of the single-commodity multicast problem, the minimum
poise subgraph can be approximated to within a factor of $O(\log k)$
with respect to the value of a natural LP relaxation in an instance
with $k$ terminals.  This is the first upper bound on the integrality
gap of the natural LP; all previous algorithms, both combinatorial and
LP-based, yielded approximations with respect to the integer optimum.
Using this integrality gap upper bound and shortest-path separators in
planar graphs, we obtain our main result: an $O(\log^3 k\log
n/(\log\log n))$-approximation for multi-commodity multicast for
planar graphs, where $k$ is the number of source-destination pairs.

We also study the minimum-time radio gossip problem in planar graphs where a message from each node must be transmitted to all other nodes under a model where nodes can broadcast to all neighbors in a single step but only nodes with a single broadcasting neighbor get a non-interfered message.
In earlier work (Iglesias et al., FSTTCS 2015), we showed a strong $\Omega(n^{\frac12 - \epsilon})$-hardness
of approximation for computing a minimum gossip
schedule in general graphs.
Using our techniques for the telephone model, we give an $O(\log^2 n)$-approximation for radio
gossip in planar graphs breaking this barrier.
Moreover, this is the first bound for radio gossip
given that does not rely on the maximum degree of the graph.

Finally, we show that our techniques for planar graphs extend to
graphs with excluded minors.  We establish
polylogarithmic-approximation algorithms for both multi-commodity
multicast and radio gossip problems in minor-free graphs.
\end{abstract}

\newpage

\setcounter{page}{1}

\section{Introduction}
\label{sec:intro}

Rumor spreading in networks has been an active research area with
questions ranging from finding the minimum possible number of messages
to spread gossip around the network~\cite{gossip0,gossip,gossip1} to
finding graphs with minimum number of edges that are able to spread
rumors in the minimum possible time in the network~\cite{brgraphs}.
There is also considerable work in the distributed computing
literature on protocols for rumor spreading and gossip based on simple
push and pull paradigms (e.g., see~\cite{feige-rumor,karp,gia2,pana1}).

The focus of this paper is the class of problems seeking to minimize
the time to complete the rumor spread, the prototypical example being
the {\bf minimum broadcast time problem} where a message at a root
node must be sent to all nodes via connections represented by an
undirected graph in the minimum number of rounds.  Under the popular
``telephone'' model, every node can participate in a telephone call
with at most one other neighbor in each round to transmit the message,
and the goal is to minimize the number of rounds. This problem has
seen active work in designing approximation
algorithms~\cite{korpel,Ravi94,guha,kortundir}.  One generalization of
broadcast is the {\bf minimum multicast time problem}: We are given an
undirected graph $G(V,E)$ representing a telephone network on $V$,
where two adjacent nodes can place a telephone call to each other. We
are given a source vertex $r$ and a set of terminals $R\subseteq
V$. The source vertex has a message and it wants to inform all the
terminals of the message. To do this, the vertices of the graph can
communicate in rounds using the telephone model. The goal is to
deliver the message to all terminals in the minimum number of rounds.

Recently,  a  more  general  demand model  called  the  multicommodity
multicast was  introduced in~\cite{nikzad2014}.   In the  {\bf minimum
multicommodity multicast  time problem},  a graph $G(V,E)$  is given
along with  a set of  pairs of nodes $P=\{(s_i,t_i)|1\leq  i\leq k\}$,
known as  demand pairs.  Each vertex  $s_i$ has a message  $m_i$ which
needs be delivered  to $t_i$. The vertices communicate  similar to the
multicast problem. The goal is to deliver  the message from each source
to  its corresponding sink in the minimum number of rounds.
Note that there is no bound on the number of messages that
can be exchanged in a telephone call.
In this  sense, the  telephone model
captures a classic information dissemination problem where the primary
communication constraint  is the  number of  connections that  a given
node can make in each round, not link bandwidth.

\subsection{Poly-logarithmic approximation for planar multicommodity multicast.}
While even  sub-logarithmic ratio  approximations have been  known for
the                minimum               time                multicast
problem~\cite{korpel,Ravi94,guha,kortundir},     the    best     known
approximation guarantees for the multicommodity case~\cite{nikzad2014}
is $\tilde{O}(2^{\sqrt{\log  k}})$ where $k$ is  the number of
different source-sink pairs.

\begin{theorem}
\label{thm:planarmc}
There is a  polynomial time algorithm for  minimum time multicommodity
multicast with $k$ source-sink pairs in $n$-node undirected planar
graphs that constructs  a schedule of length $O(OPT \log^3  k
\frac{\log n}{\log \log n})$ where  $OPT$ is the
length of the optimal schedule.
\end{theorem}
This result  to extends  in a natural  way to  bounded genus
graphs. Our results  make  critical use  of  the fact  that planar
graphs  admit small-size balanced vertex separators that  are a combination of three
shortest paths starting from any  given node~\cite{sep0}. We aggregate messages at
the paths,  move them  along the  path and then  move them  onto their
destinations  using  a  local multicast.
To break the overall multi-commodity multicast problem into recursive subproblems, we solve an LP relaxation for the overall problem and for those pairs for which the LP uses the separator path nodes in sending messages by a "large" amount, we aggregate them to the separator paths and move them along the paths.
However, to  define  this  aggregation
automatically we  need to use a linear program which requires us to
relate another lower  bound for the schedule length  that we describe
next.

\subsection{Poise and a new LP rounding algorithm.}
Suppose that  the (single-commodity) multicast problem in  a graph $G$ with  root $r$ and
terminals $R$ admits a multicast  schedule of length $L$. Consider all
the  nodes $I  \subseteq V$  in  the graph  that are  informed of  the
message from the root in the course of the schedule. For every node $v
\in I$ consider the edge through which $v$ first heard the message and
direct this edge into $v$. It is  easy to verify that this set of arcs
forms  an out-arborescence  $T$ rooted  at  $r$ and  spanning $I$.  In
particular, every node in $I$ except $r$ has in-degree exactly one and
there is a directed path from $r$ to every vertex in $I$.

\begin{definition}
Define the poise of an undirected tree $T$ to be the sum of the
diameter of the tree and the maximum degree of any node in it.  Define
the poise of a directed tree to be that of its undirected version
(ignoring directions).
\end{definition}

The discussion above of constructing  a directed tree from a multicast
schedule  implies  that the  poise  of  the  tree constructed  from  a
multicast   schedule   of   length   $L$    is   at   most   3L   (see
also~\cite{Ravi94}).  The  following lemma  gives the relation  in the
other  direction.

\begin{lemma}~\cite{Ravi94}
\label{lem:poisetomc}
Given a  tree on $n$  nodes of poise $L$,  there is a  polynomial time
algorithm  to  construct  a  broadcast scheme  of  length  $O(L  \cdot
\frac{\log n}{\log \log n})$ from any root.
\end{lemma}

Note  that a  complete $d$-ary  regular tree of depth $d$ requires  time $d^2$  to
finish multicast from the root; If the size of the tree is $n$, then $d  = O(\frac{\log  n}{\log \log  n})$. For this tree $L = O(\frac{\log  n}{\log \log  n})$ while any broadcast scheme
takes $\Omega((\frac{\log  n}{\log \log  n})^2)$ steps  showing that the multiplicative factor
is necessary.

Even  though approximation  algorithms  for minimum  poise trees connecting a root to a set of terminals were
known from earlier work~\cite{Ravi94,kortundir,guha}, their guarantees
are  with  respect to  an  optimal  (integral)  solution and  not  any
specific  LP  relaxation.
In   particular,  the  LP-based  algorithm of~\cite{Ravi94} rounds a solution to the poise LP in phases without preserving the relation of the residual LPs that arise in the phases to the LP for the poise of the whole graph. Similarly, the  LP-based  algorithm
of~\cite{guha}   solves   a  series   of   LPs   determining  how   to
hierarchically pair terminals and form the desired broadcast tree with
cost within a logarithmic factor of the integral optimum poise, but without relating the resulting tree to the LP value of the poise of the original graph.
It is not straightforward to use these methods to derive an integrality gap for the minimum poise LP, and this has remained an open problem.
Deriving an
approximation algorithm for minimum poise subgraphs for the single-commodity multicast version with a small {\em
  integrality  gap}\/ is  a critical  ingredient in  our approximation
algorithm    for   multicommodity    multicast  problem in    planar   graphs
(Theorem~\ref{thm:planarmc}). We derive the first such result.

\begin{theorem}
\label{thm:poise-round}
Given a fractional feasible solution of value $L$ to a natural linear
programming relaxation of the minimum poise of a tree connecting a
root $r$ to terminals $R$ (POISE-L LP, see
Section~\ref{sec:multiflow-round}), there is a polynomial time
algorithm to construct a tree spanning $r \cup R$ of poise $O(L \log
k)$ where $k = |R|$ and $n = |V|$.
\end{theorem}

Our LP rounding for minimum poise are based on exploiting a connection to the
theory of multiflows~\cite{Lovasz,Cherkassky,Frank-book}; this is an interesting technique in its own right that we hope
will be useful in obtaining other LP rounding results for
connectivity structures while preserving degrees and distances.

\subsection{Radio Gossip in Planar Graphs.}
Our techniques for addressing multicommodity multicast are also
applicable to radio gossip in planar graphs.  In the radio model of
communication that also occurs in rounds, a transmitting node may broadcast to multiple nodes in around
but a node may receive successfully in a given time step
only if exactly one of its neighbors transmits. The gossip problem is
a special case of the multicommodity multicast problem where the
demand pairs include all possible pairs of nodes (alternately, every
node's message must be transmitted to every other node).  The minimum
gossip problem in the radio model has been widely
studied~\cite{GasieniecPX07} but all known upper bounds involve both
the diameter and degree of the network.  In particular, for general
$n$-node graphs, there is an $\Omega(n^{\frac12 - \epsilon})$-hardness
of approximation result for computing a minimum gossip
schedule~\cite{our-fsttcs-paper}. Our next result breaks this barrier
for planar graphs (the proof and algorithm are in Section~\ref{sec:radio}).
\begin{theorem}
\label{thm:planargossip}
There is a polynomial time algorithm  for minimum time radio gossip in
an  $n$-node undirected  planar graph  that constructs  a schedule  of
length  $O(OPT \cdot  \log^2  n)$ where  $OPT$ is  the  length of  the
optimal gossip schedule.
\end{theorem}

Since  radio  broadcast from  any  node  can  already be  achieved  with
additive poly-logarithmic time overhead above the optimum~\cite{planarradio},
our algorithm for radio
gossip focuses  on gathering all  the messages  to a single  node. For
this,  we use  the path-separator  decomposition in  planar graphs  to
recursively   decompose  the   graph   and   gather  messages   bottom
up. However, the diameter of subgraphs formed by the decomposition are
not  guaranteed  to be  bounded  so  we  use a  carefully  constructed
degree-bounded  matching   subproblem  to  accomplish   the  recursive
gathering:  these techniques  adapt and  extend the  methods used  for
constructing telephone  multicast schedules~\cite{nikzad2014}  but apply them 
for the first time to the radio gathering case.

\subsection{Minor-free Graphs}

Both our results on planar graphs also naturally extend to minor-free graphs, as similar path separator results are also known for minor-free graphs~\cite{sep1}. They are detailed in Section~\ref{sec:minorfree}in the Appendix for completeness.
\begin{theorem}
\label{thm:minormc}
There is a  polynomial time algorithm for  minimum time multicommodity
multicast with $k$ source-sink pairs in $n$-node undirected $H$-minor-free
graph for a constant sized $H$ that constructs  a schedule of length $O(OPT \log^3  k
\frac{\log n}{\log \log n})$ where  $OPT$ is the
length of the optimal schedule.
\end{theorem}

\begin{theorem}
\label{thm:minorgossip}
There is a polynomial time algorithm for minimum time radio gossip in
an $n$-node undirected  $H$-minor-free
graph for a constant sized $H$  that constructs a schedule of
length $O(OPT \log^2 n)$ where $OPT$ is the length of the optimal
gossip schedule.
\end{theorem}

\subsection{Previous Work}
\textbf{Minimum time multicast in the telephone model.} Finding optimal  broadcast schedules  for trees was  one of  the first
theoretical  problems in  this setting  and was  solved using  dynamic
programming   \cite{tree}.    For   general   graphs,   Kortsarz  and Peleg~\cite{korpel} developed an  additive approximation algorithm which
uses at most  $c\cdot OPT + O(\sqrt{n})$ rounds for  some constant $c$
in an  $n$-node graph.
They also present algorithms for graphs with small balanced vertex separators
with approximation ratio $O(\log n \cdot S(n))$ where $S(n)$ is the size of the minimum balanced separator on graphs of size $n$ from the class.
The first poly-logarithmic  approximation for
minimum  broadcast time  was  achieved by  Ravi~\cite{Ravi94} and  the
current best  known approximation ratio is  $O(\frac{\log n}{\log \log
  n})$  due to  Elkin  and Korsartz~\cite{kortundir}.  The best  known
lower bound on the approximation ratio for telephone broadcast is $3 -
\epsilon$~\cite{EK-sicomp2005}.

In   his   study   of    the   telephone   broadcast  problem,
Ravi~\cite{Ravi94} introduced  the idea of finding  low poise spanning
trees to accomplish broadcast. In  the course  of deriving a  poly-logarithmic approximation,
Ravi also showed how  a tree of poise $P$ in an  $n$-node graph can be
used  to complete  broadcast  starting  from any  node  in $O(P  \cdot
\frac{\log n}{\log \log n})$ steps. His result provided an approximation guarantee with
respect  to  the optimal  poise  of  a tree  but  not  its natural  LP
relaxation that we investigate.

Guha  et al.   \cite{guha}  improved   the  approximation  factor  for
multicasting in general graphs to $O(\log  k)$ where $k$ is the number
of terminals.  The  best known approximation factor  for the multicast
problem is $O(\frac{\log k}{\log  \log k})$ \cite{kortundir}.  Both of
\cite{kortundir,guha} present a recursive  algorithm which reduces the
total number  of uninformed terminals  in each step of  the recursion,
while using $O(OPT)$  number of rounds in that  step.  In \cite{guha},
they reduce the number of uninformed terminals by a constant factor in
each  step and  so they  obtain  a $O(\log  k)$-approximation, but  in
\cite{kortundir}, the number  of uninformed terminals is  reduced by a
factor   of   $OPT$  which   gives   a   $O(\frac{\log  k}{\log   \log
  k})$-approximation due to the fact that $OPT=\Omega(\log k)$.  These
papers also imply an approximation  algorithm with factors $O(\log k)$
and  $O(\frac{\log k}{\log  \log k})$  for the  Steiner minimum  poise
subgraph problem; however, these guarantees are again with respect to the optimum
integral value for this problem and not any fractional relaxation.

For the  multicommodity multicast problem,  Nikzad and Ravi~\cite{nikzad2014}  adapt the
methods  of~\cite{dirbro0,kortundir}  to  present  an  algorithm  with
approximation ratio  $\tilde{O}(2^{\sqrt{\log k}})$  where $k$  is the
number of different source-sink pairs. They  also show that there is a
poly-logarithmic approximation inter-reducibility  between the problem
of finding  a minimum  multicommodity multicast  schedule and  that of
finding a  subgraph of  minimum generalized Steiner  poise (i.e., a subgraph
that connect all source-sink pairs, but is not necessarily connected overall, and has minimum sum of maximum degree and
maximum distance in the subgraph between any source-sink pair).

\textbf{Radio Gossip.} The radio broadcast and gossip  problems have been extensively studied
(see  the  work  reviewed in  the  survey~\cite{gasieniec2010}).   The
best-known   scheme   for  radio   broadcast   is   by  Kowalski   and
Pelc~\cite{KowalskiP} which completes in time $O(D + \log^2 n)$, where
$n$ is the number  of nodes, and $D$ is the diameter  of the graph and
is  a  lower bound  to  get  the message  across  the  graph from  any
root. The  $O(\log^2 n)$ term  is also unavoidable as  demonstrated by
Alon et al.~\cite{Alonetal} in an  example with constant diameter that
takes $\Omega(\log^2  n)$ rounds  for an  optimal broadcast  scheme to
complete. Elkin  and Korsartz~\cite{EK-sidma2005} also show  that
achieving a bound better than additive  log-squared is not possible unless  $NP  \subseteq
DTIME(n^{\log \log n})$. For planar  graphs, the best upper bound for radio
broadcast time is $D+O(\log  n)$  given  by~\cite{planarradio}.  The  best
bound for  radio gossip known so  far, however, is $O(D  + \Delta \log
n)$ steps  in an $n$-node graph  with diameter $D$ and  maximum degree
$\Delta$~\cite{GasieniecPX07}, even though there is no relation in general between the optimum
radio gossip time and the maximum degree.    Indeed, for  general   graphs,  there   is  a
polynomial inapproximability lower-bound for the minimum time
radio gossip problem~\cite{our-fsttcs-paper}.

\textbf{Planar path separators.} For  our results  on  planar graphs,  we rely  on  the structure of
path-separators.
Lipton  and Tarjan  first found small $O(\sqrt{n})$-sized separators for
$n$-node undirected planar graphs~\cite{sep}.  More recently,  planar separators  based on
any spanning tree of a planar graph were found~\cite{sep0} with the following key property: these balanced vertex separators can be formed by starting at any vertex and taking the union of three shortest paths from this vertex. Minor-free
graphs also  admit small  path-separators as found  by~\cite{sep1}; in
this case, the number of paths  used depends on the  graphs which are
excluded minors, but stays constant for constant-sized excluded minors.

\section{LP Rounding for Multicast in General Graphs} \label{sec:multiflow-round}

In this section we present an approximation algorithm for finding a
{\bf minimum poise Steiner subgraph}, and establish an LP integrality
gap upper bound, thus proving Theorem~\ref{thm:poise-round}.  We begin
by presenting a linear program for a multicommodity generalization of
minimum poise Steiner subgraph, which is useful for the multicommodity
multicast problem.  This linear program, when specialized to the case
where we need to connect a root $r$ to a subset $R$ of terminals, is
our LP for the minimum poise Steiner subgraph problem.

\subsection{Linear Program for Poise}
\label{subsec:poise-lp}
The generalized Steiner poise problem is to determine the existence of a
subgraph containing paths for every demand pair in $K = \{(s_i,
t_i)|1 \leq i \leq k\}$ of poise at most $L$, i.e. every demand pair
is connected by a path of length at most $L$ and every node in the
subgraph has degree at most $L$.

We use indicator variables $x(e)$ to denote the inclusion of edge
$e$ in the subgraph.  Since the poise is at most $L$, this is also an
upper bound on the length of the path from any terminal to the root.
 For every terminal $(s_i, t_i) \in K$, define ${\cal P}_i$ to
be the set of all (simple) paths from $s_i$ to $t_i$.  We use a variable
$y_t(P)$ for each path $P \in {\cal P}_i$ that indicates whether
this is the path used by $s_i$ to reach $t_i$ in the subgraph.  For a path
$P$, let $\ell(P)$ denote the number of hops in $P$.  The integer
linear program for finding a subgraph of minimum poise is given below.
\[
\begin{array}{lll}
\mbox{minimize} & L = L_1 + L_2  & (POISE-LP)\\
\mbox{subject to} & \sum_{e \in \delta(v)} x(e) \leq L_1 & \forall v \in V\\
 & \sum_{P \in {\cal P}(t,r)} y_t(P) = 1 & \forall t \in R\\
 &  \sum_{P \in {\cal P}(t,r)} \ell(P) y_t(P) \leq L_2 & \forall t \in R\\
 &  \sum_{P \in {\cal P}(t,r): e \in P} y_t(P) \leq x(e) & \forall e \in E, t \in R\\
 & x(e) \in \{0,1\} & \forall e \in E\\
 & y_t(P) \in \{0,1\} & \forall t\in R, P \in {\cal P}_L(t,r).
\end{array}
\]

The first set of constraints specifies that the maximum degree of any
node using the edges in the subgraph is at most $L_1$.  The second set
insists that there is exactly one path  chosen between every pair $(s_i, t_i)\in K$.  The
third set ensures that the length of the path thus selected is at most
$L_2$.  The fourth set requires that if the path $P \in {\cal P}_i$ is chosen to connect $s_i$ to $t_i$, all the edges in the path must
be included in the subgraph.

We will solve the LP obtained by relaxing the integrality constraints
to nonnegativity constraints\footnote{Even though the number of path variables is exponential,
  it is not hard to convert this to a compact formulation on the edge
  variables that can be solved in polynomial time. See e.g.,~\cite{Ravi94}}, and get an optimal solution $x, y \geq
0$.

For the remainder of this section, we will focus on the rooted version of this problem. 
In particular, there will be a root $r$ and set of terminals $R$, then we 
will make $K=\{(r,t)|t\in R\}$. It still remains to round a solution to POISE-LP
to prove Theorem~\ref{thm:poise-round}.  Before presenting the rounding
algorithm in Section~\ref{subsec:round}, we describe a result on multiflows that
will be useful in decomposing our LP solution into a set of paths that
match terminals with each other.
\subsection{Preliminaries}
\label{subsec:prelim}
Given an undirected multigraph $G$ with terminal set
$T  \subset  V$  of  nodes, a  \emph{multiflow}  is  an  edge-disjoint
collection  of paths  each  of which  start and  end  in two  distinct
terminals in $T$. The value of the multiflow is the number of paths in
the collection. Such a path between two distinct terminals is called a
$T$-path  and  a multiflow  is  called  a  $T$-path packing.  For  any
terminal $t \in T$, let  $\lambda(t,T \setminus t)$ denote the minimum
cardinality of  an edge  cut separating  $t$ from  $T \setminus  t$ in
$G$. Note that in any multiflow,  the maximum number of paths with $t$
as an  endpoint is  at most  $\lambda(t,T \setminus  t)$. Furthermore,
since every  path in a  multiflow has to  end in distinct  vertices in
$T$, the  maximum value of any  multiflow for $T$ is  upper bounded by
$\sum_{t \in  T} \frac{\lambda(t,T \setminus t)}{2}$,  by summing over
the maximum number  of possible paths from each  terminal and dividing
by two  to compensate  for counting  each path  from both  sides. This
upper bound can be achieved if a simple condition is met.
\begin{theorem}\cite{Lovasz,Cherkassky}
\label{thm:multiflow}
If every vertex in $V \setminus T$ has even degree, then there exists a multiflow for $T$ of value $\sum_{t \in T} \frac{\lambda(t,T \setminus t)}{2}$.
\end{theorem}
The  following simple  construction  will be  useful  in the  rounding
algorithm to identify good paths to merge clusters. It is based off of a lemma from~\cite{Ravi94}.
\begin{lemma}
\label{lem:forest-subgraph}
Let $G$  be a digraph where  every node has at most one outgoing edge
(and no self loops). In polynomial  time, one can find an edge-induced
subgraph $H$ of $G$  such that $H$ is a partition of  the nodes of $G$
into a forest  of directed trees each being an  inward arborescence, and with $|E(H)| \geq |E(G)|/2$.
\end{lemma}

\begin{proof}
Consider any connected component of $G$, if there are $v$ vertices, then there are either $v$ or $v-1$ edges (as each vertex has outdegree at most 1). If there are $v$ edges, there is a cycle. When we remove an edge from the cycle, we now have a connected component with $v-1$ edges. If there are $v-1$ edges and all the vertices have outdegree at most 1, then it is already an inward arborescence.

An algorithm to find $H$ would simply check the graph for directed cycles, and if any cycle exists, it would remove an edge from that cycle. Any component which had a cycle has at least 2 edges, and we remove at most 1 edge from every component. So, the resulting graph $H$ has at least half the edges that the original graph $G$ had.
\end{proof}

\subsection{The Rounding Algorithm}
\label{subsec:round}
The main idea of Algorithm~\ref{alg:poise} is to work in $O(\log k)$ phases,
reducing the number of terminal-containing components in the subgraph
being built by a constant fraction at each
stage~\cite{MBA-survey-ravi-LATIN06}.  We begin with an empty tree
containing only the terminals $R$, each in a cluster by themselves.
In each phase, we will merge a constant fraction of the clusters
together carefully so that the diameter of any cluster increases by
at most an additive $O(L)$ per phase: for this, we choose a terminal as a
center of each cluster. When we merge clusters, we partition the
clusters into stars where we have paths of length $O(L)$ from the
centers of the star leaf clusters to the center of the star
center-cluster. These steps closely follow those in~\cite{Ravi94}.
The crux of the new analysis is to extract a set of stars that merge a
constant fraction of  the current cluster centers using  a solution to
POISE-L LP.

\begin{algorithm}[h]
\caption{LP Rounding for Poise-L tree}
\begin{algorithmic}[1]
  \label{alg:poise}
  \STATE Clusters $\C \leftarrow R$; Centers $\C^* \leftarrow R$; Solution graph $H \leftarrow \emptyset$; Iteration $i \leftarrow 1$.
  \WHILE {$|\C| > 1$}
  \STATE Use Algorithm Merge-Centers($\C^*$) to identify a subgraph $F_i$ whose addition reduce the number of clusters by a constant fraction;
  \STATE $H \leftarrow H \cup F_i$; Update $\C$ to be the set of clusters after adding the subgraph $F_i$,  and update $\C^*$ to be the centers of the updated clusters based on the star structure from Algorithm Merge-Centers($\C^*$). Increment $i$.
  \ENDWHILE
  \STATE Add a path of length at most $L$ from $r$ to the center of the final cluster in $H$. Find a shortest path tree in $H$ rooted at $r$ reaching all the terminals in $R$ and output it.
\end{algorithmic}
\end{algorithm}

The key subroutine to determine paths to merge centers is presented in
Algorithm~\ref{alg:merge}.  This uses  the  multiflow packing  theorem
of~\cite{Lovasz,Cherkassky}.

\begin{algorithm}[h]
\label{algo:mc}
\caption{Merge-Centers($\C^*$) using LP solution $x$}
\begin{algorithmic}[1]
  \label{alg:merge}
  \STATE Multiply the POISE-LP solution $x$ by the least common
  multiple $M$ of the denominators in the nonzero values of $x$ to get
  a multigraph.

\STATE For every terminal $t \in \C^*$, retain the edges in the paths
corresponding to the paths in its LP-solution with nonzero value
(i.e., paths $P$ with nonzero $y_t(P)$), for a total of $M$
connectivity from $t$ to $r$. Note that the union of all the retained
edges gives connectivity $M$ from every $t \in \C^*$ to $r$ and hence
by transitivity, between each other.

\STATE Double each edge in the multigraph and use
Theorem~\ref{thm:multiflow} to find a multiflow of value $\sum_{t \in
  \C^*} \frac{\lambda(t,\C^* \setminus t)}{2} \geq \sum_{t \in \C^*}
\frac{2M}{2} = |\C^*|\cdot M $. Note that each terminal in $\C^*$ has at least $M$ paths in the multiflow.

\STATE Eliminate all the paths
in the multiflow of length longer than $4L$.

\STATE For every terminal $t$, pick one of the $M$ paths incident on
it uniformly at random and set this path to be $P_t$. If the chosen
path is eliminated due to the length restriction, set $P_t \leftarrow
\emptyset$.

\STATE Let $H$ be an auxiliary graph on vertex set $\C^*$ with at
most one arc coming out of each $t \in \C^*$ pointing to the other
endpoint of $P_t$ (or add no edge if $P_t = \emptyset$).

\STATE Apply Lemma~\ref{lem:forest-subgraph} to the subgraph of $H$
made of nodes, to get a collection $H'$ of
in-trees. For each in-tree, partition the arcs into those in odd and
even levels of the tree and pick the set with the larger number of
arcs. Note that these sets form stars originating from a set of
centers and going to a single center.  Let $H''$ denote the set of
these stars.

\STATE For each arc of the stars in
$H''$, include the path $P_t$ originating at the leaf of the star
corresponding to the arc in $H''$, and output the collection of paths.
\end{algorithmic}
\end{algorithm}

\subsection{Performance Ratio}

In   this  section,   we   prove  Theorem~\ref{thm:poise-round}.   The
performance  ratio of  the  rounding  algorithm in  the  theorem is  a
consequence  of  the following  claims,  the  first of  which  follows
directly from the path pruning in Algorithm~\ref{alg:merge}.

\begin{lemma}
\label{length-prune}
The length  of each  path output by  Merge-Centers($\C^*$) is  at most
$4L$.
\end{lemma}

\begin{lemma}
\label{lem:cluster-redn}
The expected number of paths output by Merge-Centers($\C^*$) is $\Omega(|\C^*|)$.
\end{lemma}
\begin{proof}
The main observation here is that the total number of edges in the
multigraph $G$ is at most $|\C^*|\cdot L \cdot M$.  To see this, note
that each terminal $t$ retains its flow of value $1$ in POISE-LP
corresponding to the paths with nonzero value for $y_t$. Thus in the
scaled version, it retains $M$ paths to the root, and the average
number of hops in these paths is at most $L_2 \le L$ hops, for each
terminal. Summing over all terminals, the number of edges in $G$ is at
most $|\C^*|\cdot L \cdot M$.

The total  number of paths discarded  cannot exceed $\frac{|\C^*|\cdot
  M}{4}$.  Otherwise  the paths each of length at  least $4L$ each
would need  more edges in $G$  than we started with. After discarding,
we  expect to still collect at  least $\frac1M \cdot
(|\C^*|\cdot   M  -   \frac{|\C^*|\cdot  M}{4})   =  \frac{3|\C^*|}{4}$
paths fractionally. Hence the  expected number of terminals in the  subgraph $H$ is
at least $\frac{3|\C^*|}{4}$. The set of arcs finally retained in $H''$
is at least one third of the nodes of $H$, the worst case being a path
of   two  arcs.   This   leads   to  an   expectation   of  at   least
$\frac{|\C^*|}{4}$ paths finally output to merge the initial clusters.
\end{proof}

\begin{lemma}
\label{lem:dia-inc}
The distance of  any node in a cluster to its  center increases by at  most $4L$ in
the newly  formed cluster by  merging paths corresponding to  stars in
$H''$. Thus, the diameter of any cluster in iteration $i$ is at most $8iL$.
\end{lemma}
\begin{proof}
The proof  is by induction over $i$, and is immediate by observing  that any node can  reach the new
merged cluster center  say $c$ by first following the  path to its old
center, say $t$ and then following the path $P_t$ corresponding to the
arc in $H''$  from $t$ to $c$. By  Lemma~\ref{length-prune} above, the
length of $P_t$ is at most $4L$ and the claim follows.
\end{proof}

\begin{lemma}
\label{lem:deg-inc}
The maximum degree  at any node of  $G$ induced by the  union of paths
output by Merge-Centers($\C^*$) is $O(L)$.
\end{lemma}
\begin{proof}
This is a simple consequence of the performance guarantee of rounding
the LP solution obtained for the collection of paths.  Since the paths
we found pack into the LP solution $2x$ (from the property of the
multiflow packing), the expected congestion due to the chosen random
paths on any edge $e$ is at most $2x(e)$. From the first constraint in
the LP, the expected congestion at any node due to paths incident on
it is at most $2L_1 \le 2L$, by linearity of expectation.

We apply the classic rounding algorithm of~\cite{karp+lrtvv:round}.
Since the length of each path in the collection is at most $4L$ and
the expected congestion is at most $2L$, we obtain that there is a
rounding, which can be determined in polynomial time, such that the
node congestion (degree) in the rounded solution of at most $4L$.
$4L$.  \junk{ To get a Chernoff bound~\cite{motwani1995randomized} on
  the node congestion (degree) due to the chosen paths, we need the
  number of paths to be polynomial. However, even though $M$ can be
  very large, the total number of flow paths in any basic feasible
  solution to the POISE-LP is bounded by the number of constraints
  which is $O(n^3)$. Since we have at most $n$ different terminals
  whose paths are being rounded independently, we need to union over
  $O(n^4)$ events for the tail bound~\cite{motwani1995randomized},
  giving us a solution with congestion $O(L + \log n)$.}
\end{proof}

By Lemma~\ref{lem:cluster-redn}, the number of iterations of the main
Algorithm~\ref{alg:poise} is $O(\log k)$ where $k$ is the number of
terminals.  Lemma~\ref{lem:dia-inc} guarantees that the subgraph of
the final cluster containing all the terminals has distance $O(\log k
\cdot L)$ between any pair of terminals. Since the final output is a
shortest path tree of this subgraph rooted at $r$, its diameter is
also of the same order.  Lemma~\ref{lem:deg-inc} ensures that the
total degree of any node in the subgraph of the final cluster is
$O(\log k \cdot L)$, and this is also true for the tree finally
output.  This completes the proof of Theorem~\ref{thm:poise-round}.
We can derandomize the above randomized algorithm using the standard
method of pessimistic estimators~\cite{motwani1995randomized}.

\section{Approximating multicommodity multicast on planar graphs}
\label{sec:planar}

In this section we prove Theorem~\ref{thm:planarmc}.  Let $G = (V,E)$
be the given planar graph, with $n = |V|$, and let $K = \{(s_i, t_i):
1 \le i \le k\}$ be the set of the $k$ source-destination pairs that
need to be connected.  Let $\gamma = 1/\log k$.  We given a brief
overview of our algorithm PlanarMCMulticast, which is fully described
in Algorithm~\ref{alg:planar}.

PlanarMCMulticast is a recursive algorithm, breaking the original problem into smaller problems each with at most a constant
fraction of the demand pairs in $K$ in each recursive call, thus
having $O(\log k)$ depth in the recursion.  For a given graph, the
algorithm proceeds as follows.
\begin{itemize}
\item Find a node separator composed of three shortest paths from an arbitrary vertex  \cite{sep0} to break
  the graph into pieces each with a constant fraction of the original nodes.
\item
Solve a generalized Steiner poise LP on the given pairs to identify
 demand pairs that cross the separator nodes to an extent at least $\Omega(\gamma)$.
\item
Satisfy  these  demand pairs  by  routing  their messages  from  the
sources  to the  separator, moving  the messages  along the  separator
(since  they are  shortest paths,  so this  movement takes  minimal
time) and back to the destinations, by scaling the LP values by a factor of $O(\frac{1}{\gamma})$ and using Theorem~\ref{thm:poise-round} to find a low poise tree to route to/from the separator.
\item
For the remaining demand pairs (which are mainly routed within the components after removing the separators), 
PlanarMCMulticast recurses on the pieces.
\end{itemize}
The key aspect of planarity that is used here is the structure theorem
that planar graphs contain \cite{sep0} small-size balanced vertex
separators that are a combination of three shortest paths starting
from any given node.

\begin{algorithm}[h!]
\caption{PlanarMCMulticast($G, K$)}
\begin{algorithmic}[1]
\label{alg:planar}
\STATE {\bf Base case:} When $K = \{(s_1, t_1)\}$ has one demand pair,
schedule the message  on the shortest path between  the source, $s_1$,
and destination, $t_1$.

\STATE {\bf  Separate the graph:} Define  the weight of a  node as the
number of  source-destination pairs it is  part of, and the weight of a
subset of nodes as  the sum  of their weights.  Find a
3-path separator  ${\cal P}$  of $G$, given  by shortest  paths $P_1$,
$P_2$, and  $P_3$, whose removal  partitions the graph  into connected
components each of which has weight at most half that of the graph~\cite{sep0}.

\STATE {\bf Partition the terminal  pairs:} Partition the set $K$ into
two subsets, by solving the POISE-LP.
\begin{mylowitemize}
\item
Let  $K_1$ consist  of pairs  $(s_i, t_i)$  such that  in POISE-LP, the
fraction  of the unit flow  from  $s_i$   to  $t_i$  that
intersects ${\cal P}$ is at least $\gamma$.
\item
Let $K_2 = K - K_1$ consist  of the remaining pairs, i.e. pairs $(s_i,
t_i)$ such  that in  the LP, the  fraction of the unit flow
from $s_i$ to $t_i$ that intersects ${\cal P}$ is less than $\gamma$
\end{mylowitemize}

\STATE {\bf Scale flow for pairs in $K_1$}: For each pair $(s_i, t_i)$
in $K_1$,  scale the flow between  $s_i$ and $t_i$ in  the POISE-LP by
 $\frac{3}{\gamma}$ so there  exists a path $P_j$  which intersects
a unit of this scaled $s_i$-$t_i$  flow; remove other $s_i$-$t_i$  flows that  does not
intersect $P_j$ up to a unit. Assign $(s_i,t_i)$ to a set $S_j$.

\STATE {\bf Create 3 minimum poise Steiner tree problems for $K_1$}:
For  each  path $P_j$,  create  a  minimum  Steiner poise  problem  as
follows: (i) attach, to the graph, an auxiliary binary tree $T_j$ with
nodes of $P_j$ forming the leaves, and adding new dummy internal nodes (This step is similar to~\cite{nikzad2014}); (ii)
set the root of  the binary tree to be the root  for the Steiner poise
problem, and the terminals to be all the $s_i$ and $t_i$ in $S_j$.

\STATE  {\bf Round  the POISE-LP  solution}: For  each $P_j$,
round the LP to obtain a  Steiner tree $T_j$ of small poise connecting
all the  terminals in  $S_j$ with  the root  using the  algorithm from
Theorem~\ref{thm:poise-round}.

\STATE     {\bf     Construct     schedule     for     $K_1$}:     Use
Lemma~\ref{lem:poisetomc}  on the  tree $T_j$  to perform  a multicast
between all terminals in it as  follows: use the multicast schedule to
move the messages,  from the sources, till they hit the path
$P_j$, then move messages along the path followed  by the multicast
schedule in reverse to  move them towards the destinations. (Moving 
messages along a path can be achieved by a schedule that alternates between the even and odd matchings in the path for as many steps as the target length of the schedule)

\STATE {\bf Scale flow for $K_2$}: For each pair $(s_i, t_i)$ in
$K_2$, remove any flow that intersects ${\cal P}$ and scale the
remaining flow (by a factor of at most $\frac{1}{1-\gamma}$) so as to continue to have unit total flow between the pair.

\STATE {\bf  Recurse for $K_2$}:  For each connected  component $C_j$,
let $K_2^j$ denote  the subset of $K_2$ with both  terminals in $C_j$.
Run $\mbox{ PlanarMCMulticast}(C_j, K_2^j)$ in parallel.

\end{algorithmic}
\end{algorithm}

We now prove that PlanarMCMulticast constructs,  in polynomial-time, a
multicommodity multicast schedule a schedule  of length $O((OPT \log^3 k  \cdot \frac{\log  n}{\log \log  n})$ where
$OPT$ is the length of the optimal schedule.
\begin{enumerate}
\item Observe  that $3OPT$  is  an upper  bound  for the  value  $L$ for  the
POISE-LP  for  this   instance.
\item In  Step   2  of
PlanarMCMulticast,  the separator  is obtained  using the  algorithm in
\cite{sep0}.  In  Step 3, we  use POISE-LP, as
specified  in Section~\ref{sec:multiflow-round}   to  find   the
fractional  solution. In  Step 4  of PlanarMCMulticast,  unit $s_i-t_i$
flow is achieved by scaling up the  LP cost by at most a factor of
$3/\gamma$ since at least $\gamma/3$  flow intersects one of the three
paths  in ${\cal  P}$.   Now,  observe that  this  scaled LP  solution
immediately  yields a  valid solution  to  POISE-LP in Step  5. Applying
Theorem~\ref{thm:poise-round},   in  Step   6,  with   the  value   of
$L=O(OPT\log k)$ gives a  tree of poise $O(OPT \log^2 k)$.
\item The algorithm performs $O(\log  k)$
(recursive)  phases; the  poise of  the tree  at the  $i$th level  of
recursion is itself based on an LP that has been scaled by a factor of
at most  $\frac{1}{1 -  \gamma}$ (in  Step 8) in  the previous  $i-1 =
O(\log   k)$    iterations   followed   by   a    final   scaling   of
$\frac{3}{\gamma}$ in  the last iteration.  In any iteration,  the total factor by
which the initial LP value of $OPT$ is scaled is at most $(\frac{1}{1
  - \gamma})^{i-1}     \cdot     \frac{3}{\gamma}    \leq     (1     +
\gamma)^{\frac{1}{\gamma}} \cdot  \frac{3}{\gamma} = O(\log  k)$ since
$\gamma = \frac{1}{\log k}$.

\item In Step 7, we incur a multiplicative factor of
$O(\frac{\log n}{\log \log n})$ in going  from a small poise tree to a
schedule. \junk{In  Step 7,  we can ignore  the extra $\log  n$ added  by the
virtual  tree since  applying Theorem~\ref{thm:poise-round},  with the
value of $L=O(OPT\log k + \log  n)$ (instead of $L=O(OPT\log k)$) does
not change the  poise of the integral tree obtained  from rounding.} Here,
 we crucially use the fact  that the separator  paths are
shortest paths -  for a demand pair $(s_i, t_i)$  let $f_i$ denote the
first vertex on  the separator path that the message  arrives at after
leaving source  $s_i$ and  let $l_i$  denote the  last vertex  (on the
separator  path) that  the message  departs from,  on its  way to  the
destination $t_i$;  then $f_i$ and $l_i$  must be at most  an additive
$O(OPT)$ of the sum of the lengths of the paths from  $s_i$ to $f_i$ and
$l_i$  to $t_i$ along the separator path,  since every demand  pair has a path of
length $O(OPT)$ between them in the LP solution in this subgraph. Thus in Step  7, we can wait to aggregate  all messages from the
sources at  the separator  path, then  move all  the messages  one way
along the path  and then the opposite  way, for as many  time steps as
the poise of  the integral tree, without more than  tripling the total
schedule.

\item Since there are $O(\log k)$ recursive phases, the final schedule
  has length\\ $O(OPT \log^3 k \frac{\log n}{\log
    \log n})$.
\end{enumerate}
This   proves
Theorem~\ref{thm:planarmc}.

\section{A polylogarithmic approximation for radio gossip on planar graphs}
\label{sec:radio}
In this  section, we present an  $O(\log^2 n)$-approximation algorithm
for finding a radio gossip schedule on planar graphs, and prove Theorem~\ref{thm:planargossip}.

Let $G = (V, E)$
be a given planar graph. Once the messages from all nodes have all been
gathered  together at a node we  can easily  broadcast them  back out  in $O(OPT
+ \log^2 n)$  rounds using~\cite{KowalskiP}. So we focus on gathering
the messages together at one node. To do this, we
recursively find 3-path separators in the graph~\cite{sep0} to decompose it into connected components. Then, working backwards, we gather messages
from  the 3-path separators found in an iteration at the nodes of the  $3$-path
separators found in previous iterations, using techniques from telephone multicast~\cite{nikzad2014}. The key properties used in the recursive algorithm are that
the number of paths in the separator is a constant $3$ and the paths are all shortest
paths in the component they separate from some vertex.

We  assume the  optimal  schedule has  length $L$.
Note that $L  \leq 2n$ since gossip can be achieved by simply choosing any spanning tree and  broadcasting one node at a  time in post-order (to gather all messages at the root) and then in  reverse post-order (to spread all messages back to all nodes). We also assume that $L >2$ ( $L \leq 2$ only occurs if there are 1 or 2  nodes total). The details of the algorithm
are given in Algorithm~\ref{alg:gossip}.

\begin{algorithm}[h]
\caption{A gathering procedure for radio gossip in planar graphs.}
\begin{algorithmic}[1]
  \label{alg:gossip}
  \STATE Clusters $\C_0 \leftarrow \{V\}$; Vertices $V_0 \leftarrow V$; Graph $G_0 \leftarrow G$; Iteration $i \leftarrow 1$.
	\WHILE {$V_{i-1}\neq \emptyset$}
  \FORALL {connected component $C\in \C_{i-1}$}
  \STATE Choose some $v\in C$. Find shortest paths $p_1,p_2,p_3$ from $v$ that form a 3-path separator in $C$ using~\cite{sep0}; Add these to $P_i$, the paths found in the $i$th iteration.
  \STATE Add $v$ and every $(2L+1)$st vertex along paths $p_1,p_2,p_3$ to $N_i$
  \ENDFOR
	\STATE Remove the vertices in $P_i$ from $V_{i-1}$ to get $V_i$; Let $G_i$ be $G[V_i]$ and $\C_i$ denote the connected components of $G_i$; Increment $i$.
	\ENDWHILE
  \WHILE { $i>0$}
	\STATE Do $2L$ rounds of radio broadcasts in series on nodes that are $2L+1$ apart from each other along the paths in $P_i$ to gather all the messages on $P_i$ at the nodes $N_i$.
	\STATE Form $G'_i$ a bipartite graph from $N_i$ to $U_i=\cup_{j=1}^{i-1} P_j$. Add an edge $uv\in E'_i$ if there is a path from $u \in N_i$ to $v$ in $G[\C_{i-1}\cup \{v\}]$ of length at most $L$. Find a $3L$-matching in $G'$ where every vertex of $U_i$ has degree at most $3L$.
    \STATE Do up to $L$ rounds of radio broadcast to get the messages from $N_i$ to within one node of $U_i$, along the paths in the $3L$-matching found above. Note that the messages stay within the component in $\C_{i-1}$ containing $u$ for this part.
	\STATE Move the messages from the last nodes in $\C_{i-1}$ to their destination nodes in $U_i$ in the $3L$-matching using at most $9L$ rounds for each of the paths ($27L\log n$ total).
	\STATE Decrement $i$
	\ENDWHILE
	
\end{algorithmic}
\end{algorithm}

We will first prove a couple lemmas needed for the proof of Theorem~\ref{thm:planargossip}.
\begin{lemma}
\label{lem:radiomatching}
The graph $G'_i$ has a $3L$ matching which matches ever vertex of $N_i$ to a some vertex of $U_i$ and every vertex of $U_i$ has degree at most $3L$.
\end{lemma}

\begin{proof}
Consider the  graph $G'_i=(N_i,U_i,  E'_i)$.  Let $p_v$  be a path that $v$'s message take
from $v$ to $r$ in the optimal  solution.  Let $p'_v$ be the prefix of
$p_v$ until the first vertex of $U_i$. All the paths $p_v$ have length
at most $L$ (since this is the length of the optimal schedule).
For each  node $w\in V_i$,  $w$ is in at most one of the
$p'_v$ for $v\in N_i$. This is because if two $p'_v$'s from the same path
in $P_i$ arrive at a node, there would be a path of length at
most $2L$ between two nodes in $N_i$ from the same path $P_i$;
But the paths in $P_i$ were chosen to be shortest paths in $\C_{i-1}$, and
$N_i$ were nodes that were pairwise distance $2L+1$ from each other,
a contradiction.  Now consider $u \in  U_i$ and the $p'_v$  for $N_i$ that match to $u$:
there can be  at most $L$ nodes  from which messages go from $V_i$ to
$u$ (since that is an upper bound on its message receiving degree in the optimal solution).
Thus,  in the optimal  solution there are  at most $3L$  paths from
$N_i$ to  any specific node  in $U_i$ and  these paths have  length at
most $L$. So,  there must exist a $3L$-matching in  $G'$ which matches
every vertex of $N_i$  to some vertex of $U_i$ and no vertex of $U_i$
has degree more than $3L$.

\end{proof}

\begin{lemma}
\label{lem:radioloop}
Each iteration of the step 13 in algorithm~\ref{alg:gossip} takes time at most $O(L \log n)$ and moves the messages to $U_i$.
\end{lemma}

\begin{proof}
In Algorithm~\ref{alg:gossip}, step 13 is just  to achieve  the last  step of  message movement in the $q_v$
paths. Each node $w\in V_i$ can be  adjacent to at most $3$ nodes in a
given path of $P_j$ for any $j < i$, as these paths are shortest paths in $G_{j-1}$,
and in  particular these  nodes are  within $2$ of  each other  in the
path. Also,  for a given $w\in  V_i$, $w$ can be  adjacent to multiple
paths  in  $P_j$  but they must  all  be  in  the same  component  of
$\C_{j-1}$, and there is at most three such paths.
Let $S_j^k(\ell)$  be  every third  vertex  on the  paths
$P_j^k$  starting with  the  $\ell$th  vertex. In  $3L$  steps (the maximum
degree of the matching at these nodes) we  can
gather the messages that need to be received at $S_j^k(\ell)$ as no
node is  adjacent to two nodes in this set. Doing this gathering for
every shift  $\ell$ from one to three, and each of three choice of which path $k$,  a total of
$27L$  steps  gets  the  message  from  the  $q_v$  to
$P_j$. This process is  repeated for each collection $P_j$ with $j  < i$. Now all
the   messages   that   were   along  $P_i$   have   been   moved   to
some node in $U_i=\cup_{j=1}^{i-1} P_j$ in $27L\log n$ steps.

\end{proof}

Having established the lemmas, we now give the proof of Theorem~\ref{thm:planargossip}.

\begin{proof}[Proof of Theorem~\ref{thm:planargossip}]

First, we establish that the  algorithm runs correctly. Let  $r$ be
the root (chosen in the first iteration). First
the algorithm  gathers the  messages on  the $P_i$  to $N_i$.  We will
divide  the paths  into $P_i^1,  P_i^2, P_i^3$,  so that each component of
$\C_{i-1}$ that has three paths puts one path in each of these sets. Now, we
will handle each of  the $P_i^j$ one at a time.  To deal with $P_i^j$,
in the $k$th step, the nodes which are  $2L-k$ further from a node in $N_i$ broadcast
all their messages to the node one closer to $N_i$. There is no interference amongst the
nodes in  $V_{i-1}$ as  only nodes  at distance at  least $2L$  in each
component of $\C_{i-1}$  are broadcasting at a time, and nodes in different components are non-adjacent (since they are disconnected by the separators $P_{i-1}$). This will gather
all  the messages  along  $P_i^j$. Doing  this once  for  each of  the
$P_i^j$ in $2L$ steps, all the messages currently  on $P_i$
will be gathered at the $N_i$ in $6L$ steps. The messages are all currently at $N_i$
or $U_i$.

Lemma~\ref{lem:radiomatching} tells us that $G'_i$ has a $3L$-matching as desired, and so we can find such a matching. Once  a matching  is  chosen,  let $q_v$  be  a  shortest path  within
$G_{i-1}$  from  $v\in  N_i$  to  the  vertex  it is matched  with  in
$U_i$. Within  each component  of $\C_i$, we  can broadcast  along the
$q_v$ for  every  $N_i$ in  one of  the paths  simultaneously;
There will  be no interference as the $N_i$'s and their matching paths are far apart within
$\C_{i-1}$. Thus, it  takes at most $L$ rounds of  radio broadcasting to move
the messages from $N_i$ along their $q_v$ to the vertex before $U_i$.

Lemma~\ref{lem:radioloop} gives us that in time $O(L\log n)$ we move the messages onto $U_i$.

In the last iteration of the process, we will have all the messages on
the first path separator $P_1$. $P_1$  is a path separator of shortest
paths on the whole  graph and the diameter of $G$ is  a lower bound on
$L$. So, in $3L$ steps we can  move the messages from $P_1$ to $r$. We
have  now  successfully gathered  all  the  messages  to $r$.

The time it takes to deliver all the messages to $r$ is at most $O(L \log^2 n)$. The  path
separator ensures that each component has at most a constant fraction of the
number of vertices of the original graph.  Therefore, the final $i  \leq \log  n$.  Each iteration,  the number  of
rounds of broadcast we  do is $6L$ in the first part  and $27L \log n$
in the  last step.  So, this  schedule uses $O(L  \log^2 n)$  steps to
gather all the messages at $r$.
\end{proof}

\section{Minor-free graphs}
\label{sec:minorfree}

In this section, we prove that both results on planar graphs can be extended to any family of minor-free graphs. For this section, we will have to use the more general definition for path-separators.

\begin{definition}~\cite{sep1}
A vertex-weighted graph $G$ is $k$-path separable if there exists a subgraph $S$, called a $k$-path separator, such that:
\begin{enumerate}
\item $S=P_0\cup P_1\cup\dots$ where each subgraph $P_i$ is the union of $k_i$ shortest paths in $G\backslash \cup_{j<i} P_j$
\item $\sum_i k_i \leq k$
\item either $G\backslash S$ is empty, or each connected component of $G\backslash S$ is $k$-path separable and has total vertex weight at most half of the original.
\end{enumerate}
\end{definition}

This definition of path separator while more complicated can be integrated into our PlanarMCMulticast algorithm and algorithm~\ref{alg:gossip} with only small adjustments. We will use the following to theorem which tells us when a $k$-path separator exists and can be found. 

\begin{theorem}~\cite{sep1}
\label{thm:pathsep}
Every $H$-minor-free weighed connected graph is $k(H)$-path separable, and a $k(H)$-path separator can be computed in polynomial time.
\end{theorem}
Note that $k(H) = O(|H|)$ in the above construction.

\subsection{Multicommodity multicast in minor-free graphs}
Consider that our graph is $H$-minor free and $k(H)$ is the number of paths needed for the path separator. We will need to  repeat steps 3-8 of the original algorithm for each subgraph of shortest paths. Other than that, the algorithm remains the same.

The main changes to the algorithm are as follows. 
\begin{itemize}
\item The path separator we find is different.
\item We may need to iterate through steps 3-8 of the original algorithm multiple times.
\end{itemize}

Since these are the only changes to the algorithm, we first see that finding the path separator takes polynomial time, and the number of iterations increases by at most a factor of $k(H)$ (a constant) so the algorithm still runs in polynomial time.

The only other change to the analysis is that the number of recursions is $k(H) \log k$. This gives that we incur a factor of $(\frac{1}{1-\gamma})^{k(H)\gamma}\frac{3}{\gamma}$ when scaling the LP. Since $\gamma=\frac{1}{\log k}$, the LP gets scaled by a factor of $O(e^{k(H)}\log k)$. Since $|H|$ and hence $k(H)$ is a constant, we get that the resulting schedule from this algorithm is at most  $O((OPT  \log k  + \log n)\log^2   k    \frac{\log   n}{\log    \log   n})$. In terms of $k(H)$, our algorithm builds a schedule that takes $O((OPT  \log k  + \log n)\log^2   k    \frac{\log   n}{\log    \log   n} k(H)e^{k(H)})$

This proves theorem~\ref{thm:minormc}.

\subsection{Radio Gossip in Minor-free Graphs}
The modification to go from radio gossip on planar graphs to radio gossip on minor-free graphs is even more simple. Let $k=k(H)$ where our graph $G$ is $H$-minor-free. We only change how we set up all of our initial sets; $P_i, C_i, V_i, N_i, U_i$.  For each iteration, we will define up to $k$ of these sets one for each of the (potential) subgraphs which compose the path separator.

The only major change to this algorithm is the set-up. In the set-up, we have to process the subgraphs which compose the path separator one at a time (as opposed to there only being one subgraph which is the whole path separator). This only increases the number of iterations by a factor a $k(H)$.

The other change is that everywhere we had a $3$ before it now becomes a $k=k(H)$. All our previous lemmas and theorems hold if we change the $3$ arising from the planar case to $k$. Therefore, we have shown that this algorithm produces a schedule for gathering which runs in time $O(L \log^2 n)$ (or $O(L k^3 \log^2 n)$ if $k$ is not constant) to gather all the messages in one place.

We again use the result of Kowalski and Pelc to broadcast the messages once they have been gathered~\cite{KowalskiP}. This broadcast takes time $L+O(\log^2 n)$, so the whole gossip schedule takes time $O(L\log^2)$ proving Theorem~\ref{thm:minorgossip} as desired.  
We again use the result of Kowalski and Pelc to broadcast the messages once they have been gathered~\cite{KowalskiP}. This broadcast takes time $L+O(\log^2 n)$, so the whole gossip schedule takes time $O(L\log^2 n)$ proving Theorem~\ref{thm:minorgossip} as desired.  If we don't assume $k=k(H)$ is a constant, then the algorithm takes time $O(Lk^3\log^2 n)$. 

\bibliography{refs}

\end{document}